\documentclass[11pt, oneside]{amsart} 

\usepackage{geometry,amsmath,mathtools}                		
\usepackage{graphicx}				
	\usepackage[nocompress]{cite}

\usepackage{amsaddr}
\usepackage{amssymb}
\usepackage{enumitem,xcolor,tikz}
\usepackage{epigraph}
\usepackage{ragged2e}
\usetikzlibrary{arrows.meta,bending,automata}
\newcommand{\define}[1]{{\bf \boldmath{#1}}}
\theoremstyle{definition}
\newtheorem{thm}{Theorem}
\newtheorem{definition}[thm]{Definition}
\newtheorem{lemma}[thm]{Lemma}
\newtheorem{proposition}[thm]{Proposition}

\newtheorem{remark}[thm]{Remark}

\newtheorem*{terminology}{A remark about terminology}

\newcommand{\relGraph}{\mathrm{relGraph}}
\newcommand{\krelGraph}{ \prescript{}{k}{\relGraph}}
\newcommand{\semigroup}{\mathrm{Semigroup}}

\newcommand{\ob}{\mathrm{ob}}
\newcommand{\SG}{\mathrm{SG}}
\newcommand{\RR}{\mathrm{R}}
\newcommand{\Adv}{\mathrm{Adv}}
\newcommand{\Fr}{\mathrm{Fr}}
\newcommand{\Work}{\mathrm{Work}}

\title[A unified framework for equivalences in social networks (Working paper)]{A unified framework for equivalences in social networks}

\author{Nina Otter and Mason A. Porter} 

\address{Department of Mathematics, University of California Los Angeles \\ 
Los Angeles, CA 90095, USA}

\date{}							

\begin{document}
\maketitle
\thispagestyle{empty}

\begin{abstract}

 A key concern in network analysis is the study of social positions and roles of actors in a network. The notion of ``position'' refers to an equivalence class of nodes that have similar ties to other nodes, whereas a ``role'' is an equivalence class of compound relations that connect the same pairs of nodes. An open question in network science is whether it is possible to simultaneously perform role and positional analysis. Motivated by the principle of functoriality in category theory we propose a new method that allows to tie role and positional analysis together. We illustrate our method on two well-studied data sets in network science.

\end{abstract}


\section{Introduction}\label{S:intro}

Networks have been used to model social systems since the 1930s; in these models nodes
correspond to social actors, and ties to  social relations between
them\cite{freeman}. Such networks are usually called ``social networks'', and many of the techniques developed to study social networks  now constitute
some of the main methods used in network analysis at large. 

One of the major aims of social network scientists is the study of social positions and roles. 
  The notions of ``position'' and ``role'' are sometimes  used interchangeably in the literature, however here we use the  following convention, which is also adopted in, e.g., \cite{WF94}:
roughly, the notion of ``position'' refers to a collection of actors, represented by nodes,
who have similar ties to other actors, whereas a
``role'' is a collection of compound  ties that connect the same pair of actors or positions.  

In positional analysis, one  seeks to  partition the set of nodes of a social network according to a chosen measure of similarity. In practice, nodes are rarely perfectly similar, thus one needs to choose a quantitative measure  of ``approximate'' similarity. Given a partition of the nodes  in blocks of nodes that are approximately similar, there are several different methods that network scientists use to decide how to assign ties between the blocks of the partition. Using one of these methods, one obtains a network in which nodes correspond to the blocks in the partition, and the ties  capture some of the information given by the ties in the original network. The network obtained from  a partition with $n$ blocks is usually called an {\bf $n$-blockmodel} of a network.

In role analysis, given a network $G$ one studies the semigroup $\SG(G)$ generated by the relations and compound relations between the nodes of $G$, called the {\bf semigroup of relations} of $G$. The semigroup of relations typically has a large number of elements, and one is thus interested in finding  semigroups $S$ with fewer elements that are homomorphic images of $\SG(G)$, that is, such  that there is a surjective semigroup homomorphism $\SG(G)\twoheadrightarrow S$. This amounts to creating a partition on the set of relations; the resulting semigroup $S$ is called a {\bf homomorphic reduction} of $\SG(G)$.

Ideally, one would want to find a way to combine these two types of analysis. This would require being able to partition the nodes into blocks of approximately similar nodes and to partition the relations at the same time. There has been a lot of effort in this direction, see, e.g., \cite{BP86}, but this remains an open question in the field.

In our work we first consider  the ideal case of  blockmodels in which nodes within a block are all perfectly equivalent. 
In positional analysis it is often the case that one studies different types of blockmodels associated to a multirelational graph. One can think of different  blockmodels as providing an analysis of the structure of a network at different resolution levels. As was observed by Peixoto \cite{Peixoto14}, the methods that are widely used to obtain such blockmodels might not be able to distinguish between significant features in the structure of the network and noise. Thus, ideally one would want to study a nested sequence of blockmodels that captures information of the network at different resolution levels. Peixoto introduced a stochastic model to perform such an analysis \cite{Peixoto14} (see also Section \ref{S:SBM}).

 In the ideal case of perfect equivalences, given a network $G$ and a nested sequence of such blockmodels, one would obtain a sequence $G_1,\dots , G_p$ of networks and surjective network homomorphisms $G\twoheadrightarrow G_i$ for each $i$. 
Because the different partitions are nested within each other, there are also surjective network homomorphisms between the different blockmodels, and (up to a relabeling of the indices of the blockmodels) one gets a sequence $G=G_0,G_1,\dots , G_p$ of networks and network homomorphisms $\pi_{i,j}\colon G_i\twoheadrightarrow G_j$ for $i\leq j$ such that $\pi_{k,j}\circ \pi_{i,k}=\pi_{i,j}$ whenever $i\leq k \leq j$. If we now associate each network $G_i$ with its semigroup of relations $\SG(G_i)$, we have 
that each network homomorphism $\pi_{i,j}$ also induces a homomorphism $\SG(\pi_{i,j})\colon \SG(G_i)\to \SG(G_j)$ of semigroups in a way that preserves composition of homomorphisms. That is, $\SG(\pi_{i,j})=\SG(\pi_{k,j})\circ \SG(\pi_{i,k})$ whenever $i\leq k \leq j$. In the language of category theory, we have that  the assignment of the semigroups of relations to a network is a functor. This is important, because functoriality allows to investigate questions related to robustness, see, e.g., \cite{BS14}, where functoriality is the key ingredient in proving stability results, or \cite{C09} for a more general explanation of why functoriality is important in data analysis.

In the case of approximate equivalences the situation is more complex, as there is in general neither a network homomorphism from a network to any of its blockmodels, nor is there a sensible way to relate different blockmodels of the same network, or different reduction semigroups, nor does the assignment of a blockmodel to a graph induce a homomorphism on the corresponding semigroups of relations. Here we propose a 
framework that addresses all these points: we consider weighted social networks, with weights on ties capturing information about the degree of ``approximate'' similarity, and introduce a new algebraic structure, the ``truncated semigroup of relations'', which allows to perform role analysis for weighted social relations. We illustrate the methods that we introduce with two well-known data sets in social network analysis. 
 With this new framework we thus open the avenue 
to investigating robustness problems, and we also give a way of tying role and positional analysis together.

\begin{terminology}
Many of the notions that we  discuss in this manuscript are known under several different names in the social network analysis or network analysis literature. To make our paper as accessible as possible, whenever possible we mention alternative names and give appropriate references. 
\end{terminology}

%
%

\section{Multirelational graphs}

A multirelational graph consists of a set of nodes, together with a set of ties which can either be directed or undirected, and are labeled by relations. We give the definition of multirelational graph using adjacency matrices.

\begin{definition}
A \define{multirelational graph} $G$ consists of a pair $\left(V,\{A_i\}_{i=1}^r \right)$, where $V$ is a finite set of \define{vertices} of $G$
and the adjacency matrices $A_1,\dots , A_r$ are 
$n\times n$ matrices, with Boolean entries (i.e., $0$ and $1$), that we call \define{relation matrices} of $G$.
\end{definition}

\begin{remark}
Relation matrices are often called ``sociomatrices'' in the social network analysis literature, for instance in \cite{WF94}, while multirelational graphs are often also called ``multiplex networks'' in network science.
\end{remark}




In Figure \ref{F:running ex}, we give an example (which we take from \cite{WF94}) of a fictitious multirelational graph with two relation matrices, $A_H$ and $A_L$. The nodes of such a graph may represent members of a family, employees in a firm, or something else. The relation matrix $A_H$ could encode the relation ``is the parent of'' or ``oversees the work of'', and the relation matrix $A_L$ could encode the relation ``is in the same generation as'' or ``trades jobs with''.

\begin{figure}[h!]
\begin{alignat}{2}
\notag & i)\qquad\qquad&&\notag
A_H=\begin{pmatrix}
0 & 1 & 1 & 0 & 0& 0\\
0 & 0 & 0 & 1 & 0 & 0\\
0 & 0 & 0 & 0 & 1 & 1\\
0 & 0 & 0 & 0 & 0 & 0 \\
0 & 0 & 0 & 0 & 0 & 0 \\
0 & 0 & 0 & 0 & 0 & 0 
\end{pmatrix} \,,
\qquad
A_L=\begin{pmatrix}
1 & 0 & 0 & 0 & 0& 0\\
0 & 1 & 1& 0 & 0 & 0\\
0 & 1 & 1& 0 & 0 & 0\\
0 & 0 & 0 & 1 & 1 & 1 \\
0 & 0 & 0 & 1 & 1 & 1 \\
0 & 0 & 0 & 1 & 1 & 1
\end{pmatrix}\\
\notag &\\
\notag &\\
\notag &
ii)  \qquad\qquad
\notag &&H\qquad
\begin{tikzpicture}[scale=0.6,every node/.style={scale=.68}]
 \node[shape=circle,fill=black] (1) at (0,0) {};
  \node[shape=circle,fill=black] (2) at (-2,-2) {};
 \node[shape=circle,fill=black] (3) at (-2,-4)  {};
 \node[shape=circle,fill=black] (4) at (2,-2)  {};
 \node[shape=circle,fill=black]  (5) at (1,-4)  {};
 \node[shape=circle,fill=black]  (6) at (3,-4)  {};
\draw[->]
(1) edge (2)
(2) edge (3)
(1) edge (4)
(4) edge (5)
(4) edge (6);
\draw[->,color=white]
(3) edge [bend right=90] (6);
\end{tikzpicture}
\qquad\qquad
L \qquad\begin{tikzpicture}[scale=0.6,every node/.style={scale=.68}]
 \node[shape=circle,fill=black] (1) at (0,0) {};
  \node[shape=circle,fill=black] (2) at (-2,-2) {};
 \node[shape=circle,fill=black] (3) at (-2,-4)  {};
 \node[shape=circle,fill=black] (4) at (2,-2)  {};
 \node[shape=circle,fill=black]  (5) at (1,-4)  {};
 \node[shape=circle,fill=black]  (6) at (3,-4)  {};
\path[->]
(1) edge [out=225, in=-45, looseness=4,loop] (1)
(4) edge [out=45, in=315, looseness=2,loop] (4)
(5) edge [out=120, in=65, looseness=3,loop] (5)
(2) edge [out=225, in=135, looseness=4,loop] (2)
(3) edge [out=225, in=135, looseness=4,loop] (3)
(6) edge [out=45, in=315, looseness=2,loop] (6)
(2) edge [bend left] (4)
(4) edge [bend left] (2)
(3) edge [bend left] (5)
(5) edge [bend left] (3)
(6) edge [bend left] (5)
(5) edge [bend left] (6)
(6) edge [bend left=60] (3)
(3) edge [bend right=90] (6);
\end{tikzpicture}
\end{alignat}
\caption{(i) Two relation matrices, $A_H$ and $A_L$, of a multirelational graph with six vertices. (ii) Associated graphs that illustrate these relations. (We draw this example from \cite{WF94}.)}\label{F:running ex}
\end{figure}
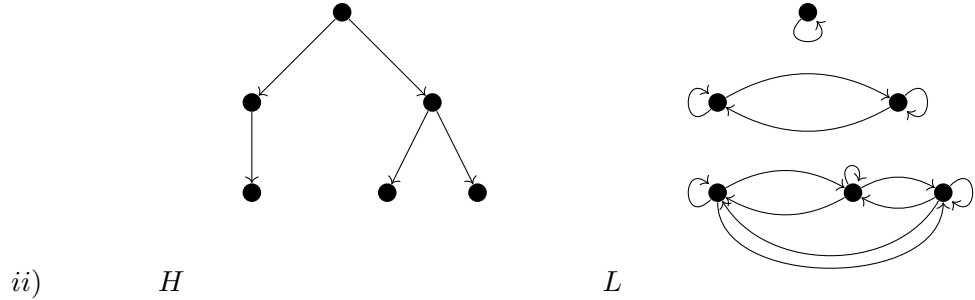


\section{Positional analysis}

When conducting positional analysis, social scientists seek to associate
a multirelational graph $G$ with a multirelational graph $G'$ with fewer nodes, such that the nodes of $G'$ represent equivalence classes of nodes of $G$ that are ``similar'' in an appropriate sense. 
Positional analysis typically includes the following steps:
\begin{enumerate}[label=(\roman*{})]
\item choosing a notion of similarity;
\item defining a quantitative measure of this similarity, to account for  ``approximate'' similarities;
\item partitioning the set of nodes of $G$ according to this quantitative measure of similarity; 
\item associating to each relation matrix $A$ a ``density matrix'' $\widetilde{A}$ that is smaller in size and in which rows and columns correspond to the blocks in the partition;
\item transforming each density matrix into a Boolean matrix called ``image matrix'' to obtain a multirelational graph $\widetilde{G}$ whose nodes correspond to blocks of the partition.
\end{enumerate}


In the following subsections,
we give a brief overview of these steps.


\subsection{Equivalences}

There are many notions of similarity that have been studied in network analysis. The most well-studied one is ``structural'' equivalence. 

\begin{definition}
Given a multirelational graph $(V,\{A_i\}_{i=1}^r)$, we say that two nodes $i$ and $j$ are \define{structurally equivalent} if 
\begin{align*}
	A_s(i,k) \ne 0 &\Longleftrightarrow \, A_s(j,k)\ne 0  \,,  \notag \\
	A_s(k,i) \ne 0 &\Longleftrightarrow \, A_s(k,j)\ne 0 
\end{align*}
for all relations $s \in \{1,\dots, r\}$

\end{definition}

Therefore, two nodes are structurally equivalent exactly when they have the same ties to all other nodes for all relations. Other types of well-studied equivalence relations include regular equivalence and isomorphic equivalence. See \cite{WF94} for an overview.


In practice, nodes are rarely perfectly structurally equivalent. For instance, in the multirelational graph in Figure \ref{F:running ex}, if there were only one relation to consider at a time, 
we would obtain the partition 
\[
	\{1\} \, , \{2\} , \{3\} , \{4\} , \{5,6\}
\]
for the relation matrix $A_H$ and the partition
\[
	\{1\}, \{2, 3\}, \{4, 5, 6\}
\]
for the relation matrix $A_L$. However, no two pairs of distinct nodes of the graph have exactly the same ties under {\em both} relations. Therefore, the partition of the graph that is given by structural equivalence is the partition in which every block is a singleton:
\[
	\{1\} \, , \{2\} , \{3\} , \{4\} , \{5\}, \{6\} \, .
\]

Consequently, in practice, researchers are interested in quantitative measures that capture when two nodes are approximately equivalent. 

%
%
%



\subsection{Measures of approximate equivalence}\label{SS:approx eq}
Several measures to quantify structural equivalence have been proposed in the literature, such as Euclidean distance, or cosine similarity. Quantitatives measures for other types of equivalences are less well studied in the literature. We refer the reader to the overviews in \cite[9.4]{WF94} and \cite[7.6]{N18}.


\subsection{Density matrices}\label{SS:density matrices}

After partitioning the set of nodes using a measure for approximate equivalence, we
replace the relation matrices with matrices of smaller size that store information about the fraction of ties 
between the blocks of the partition.

\begin{definition}\label{D:density matrix}
Given a multirelational graph $(V,\{A_i\}_{i=1}^r)$ and a partition $B_1,\dots B_m$ of $V$, the \define{permuted matrix} of the relation matrix $A_s$ is the matrix that we obtain from $A_s$ by permuting columns and rows such that rows and columns are adjacent exactly when they correspond to nodes in the same block of the partition.
 The \define{density matrix} of the relation matrix $A_s$ is the $m\times m $ matrix $\widetilde{A_s}$ with entries
\[
	\widetilde{A_s}(i,j)=\begin{frac}{ \text{number of $1$ entries in the $(i,j)$th block of the permuted matrix}}{\text{number of entries in the $(i,j)$th block}}\end{frac} \, .
\]

\end{definition}

\medskip
\noindent
We give an example of density matrices in Figure \ref{F:density matrices}.

\begin{figure}[h!]

$\widetilde{A_H}=\begin{pmatrix}
0 & 1 & 0 \\
0 & 0 & 0.5\\
0 & 0 & 0
\end{pmatrix}\,,
$\qquad\qquad
$\widetilde{A_L}=\begin{pmatrix}
1 & 0 & 0 \\
0 & 1 & 0\\
0 & 0 & 1
\end{pmatrix}
$
\caption{Using  a method for approximate equivalence  yields
the partition $B_1=\{1\}, B_2=\{2, 3\}, B_3=\{4, 5, 6\}$ of the multirelational graph from Figure \ref{F:running ex}. In this figure, we give the corresponding density matrices.}\label{F:density matrices}
\end{figure}

If the nodes in the blocks of a partition are perfectly equivalent, the entries in the density matrix are either $1$ or $0$. As we have already noted, this is rare in practice, so entries in a density matrix can be any real number in the interval $[0,1]$. We can interpret the entries in the density matrix as weights that are associated with the ties between the positions: values that are closer to $0$ indicate weaker ties between blocks, and values that are closer to $1$ indicate stronger ties between blocks.


\subsection{Image matrices}

Finally, given a density matrix, we seek to associate a Boolean matrix with it to obtain a multirelational graph $G'$ whose nodes correspond to the blocks in the partition. There are several methods to do this; see the overview in \cite{WF94}. 
One of the methods discussed in \cite{WF94}  is the $\alpha$-density criterion: given a number $\Delta\in (0,1)$, we substitute an entry $c$ of the density matrix with a $1$ if $c\geq \Delta$ and we otherwise place a $0$. The resulting matrix is called an \define{image matrix}. A common value of $\Delta$ is the fraction of $1$ entries in the relation matrix. We illustrate this choice with an example in Figure \ref{F:image matrices}.

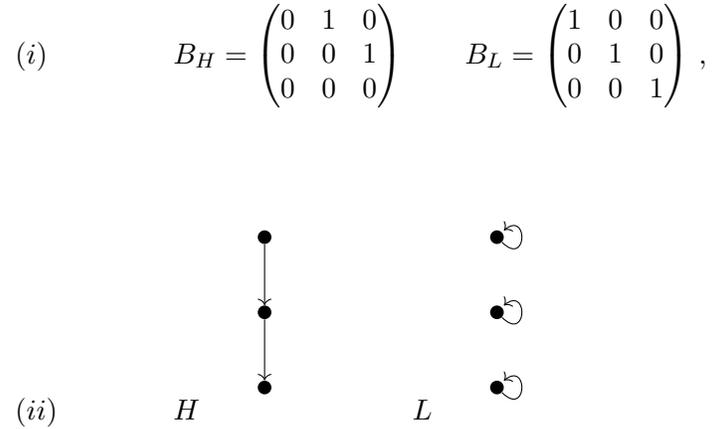
\begin{figure}[h!]
\begin{alignat}{2}
\notag & (i)\qquad\qquad&&\notag
B_H=\begin{pmatrix}
0 & 1 & 0 \\
0 & 0 & 1\\
0 & 0 & 0
\end{pmatrix}
\qquad
B_L=\begin{pmatrix}
1 & 0 & 0 \\
0 & 1 & 0\\
0 & 0 & 1
\end{pmatrix} \,, \\
\notag &\\
\notag &\\
\notag &
(ii)  \qquad\qquad
\notag &&H\qquad
\begin{tikzpicture}[scale=1,every node/.style={scale=.5}]
 \node[shape=circle,fill=black] (1) at (0,0) {};
 \node[shape=circle,fill=black] (2) at (0,-1) {}; 
 \node[shape=circle,fill=black] (3) at (0,-2) {}; 
 \path[->]
 (1) edge (2)
 (2) edge (3);
  \path[->,color=white]
 (3) edge [out=-45, in=45, looseness=4,loop] (3);
 \end{tikzpicture}
\qquad\qquad
L \qquad 
\begin{tikzpicture}[scale=1,every node/.style={scale=.5}]
 \node[shape=circle,fill=black] (1) at (0,-0.5) {};
 \node[shape=circle,fill=black] (2) at (0,-1.5) {}; 
 \node[shape=circle,fill=black] (3) at (0,-2.5) {}; 
 \path[->]
(1) edge [out=-45, in=45, looseness=4,loop] (1)
(2) edge [out=-45, in=45, looseness=4,loop] (2)
(3) edge [out=-45, in=45, looseness=4,loop] (3);
 \end{tikzpicture}
\end{alignat}

\caption{(i) The image matrices that are associated with the density matrices from Figure \ref{F:density matrices} and (ii) their corresponding graph representations. }\label{F:image matrices}
\end{figure}

\section{Role analysis}\label{S:role analysis}

In role analysis, researchers
are interested in studying all possible composite relations that arise between social actors.
One starts by defining what it means to compose two relations.

\begin{definition}
Given a multirelational graph $(V,\{A_i\}_{i=1}^r)$, we say that there exists a \define{compound relation} from node $i$ to node $j$ if there are nodes $k_1,\dots , k_l$ such that the product 
\[
	A_{s_l}(k_l,j)A_{s_{l-1}}(k_{l-1},k_l)\dots A_{s_2}(k_1,k_2)A_{s_1}(i,k_1)
\]
of the matrix entries is nonzero, where the indices $s_1,\dots s_l$ are in $
\{1,\dots ,r\}$.
\end{definition}

For instance, in the multirelational graph in Figure \ref{F:running ex}, there is a compound relation from $1$ to $3$ given by $A_{L}(2,3)A_H(1,2)$. If we interpret the nodes as representing the employees of a firm, this compound relation may express something like ``employee 1 oversees an employee who trades jobs with employee 3''.

To study all possible compounds relations, one computes the semigroup that is generated by the relation matrices of a multirelational graph.

\begin{definition}
Given a multirelational graph $(V, \{A_i\}_{i=1}^r)$, its \define{semigroup of relations} is the semigroup that is generated by the matrices $A_1,\dots , A_r$, with a binary operation given by Boolean matrix multiplication. We denote such a semigroup by $\SG(A_1,\dots , A_r)$.
\end{definition}

In Table \ref{sg table}, we give the multiplication table for the semigroup that is associated with the multirelational graph from Figure \ref{F:running ex}.

\begin{table}
\begin{tabular}{ c || c | c | c | c |}
$\circ$         & H                     & L                     & HL                    & HH                    \\ \hline\hline
 H 	&   HH   &  HL & HH &    0 \\ \hline
 L & HL & L & HL & HH  \\ \hline
 HL & HH & HL & HH & 0 \\\hline
 HH & 0 & HH & 0 & 0\\
\end{tabular}
\caption{Multiplication table for the semigroup that is associated with the  multirelation graph from Figure \ref{F:running ex}. For simplicity, we indicate the matrix $A_H$ by $H$ and the matrix $A_L$ by $L$.}\label{sg table}
\end{table}

In applications, the semigroup of relations often contains a number of elements that is too large for researchers to be able to  easily identify  equalities (or ``approximate'' equalities) between compound relations from the associated multiplication table. Therefore it is desirable to simplify the semigroup of relations by associating to it a semigroup with fewer elements, called its ``homomorphic reduction''.

\begin{definition}
Given a semigroup $S$, a \define{homomorphic reduction} of $S$ is a semigroup $S'$ together with a surjective semigroup homomorphism $S\to S'$.
\end{definition}

A homomorphic reduction of a semigroup $S$ is not unique, and different methods have been suggested for
different ways of choosing one \cite{BM80}.
In our work we introduce a way to tie the two analyses together in such a way  that a choice of blockmodel  
also gives a choice of homomorphic reduction.


\section{Combining role and positional analysis}\label{S:role positional}

There are usually two ways in which researchers have combined role and positional analysis:
\begin{enumerate}[label=(\roman*{})]
\item first partition the set of nodes of a multirelational graph, and then study patterns between ties in the reduced multigraph;
\item first study the patterns of the ties in the (unreduced) multirelational graph, and then partition the set of nodes.
\end{enumerate}

%


We illustrate these two types of approaches with the diagram in Figure \ref{F:pos and role analysis}.
Performing the two typically leads to different results. That is,
the diagram in Figure \ref{F:pos and role analysis} does not commute.
Ideally, it is desirable 
to have a way 
to simultaneously perform positional and role analysis. In such an approach, there is no need to choose a path in the diagram in Figure \ref{F:pos and role analysis}. A key contribution of our paper is to recast this problem under a different angle: instead of seeking to perform the two analyses simultaneously, we focus on how performing one analysis affects the other, and we argue that this question is crucial to performing role and positional analysis for real-world networks, and to investigate stability questions. 

%
%
%
%


\begin{figure}
\begin{tikzpicture}
\node (1) at (0,0) {multirelational graph};
\node (2) at (4,0) {roles};
\node (3) at (0,-3) {positions}; 
\node (4) at (4,-3) {roles and positions};
\draw[->] (1) edge (2)
(2) edge (4)
(1) edge (3)
(3) edge (4);
\end{tikzpicture}
\caption{A schematization of the two types of ways of combining positional and role analysis.}\label{F:pos and role analysis}
\end{figure}
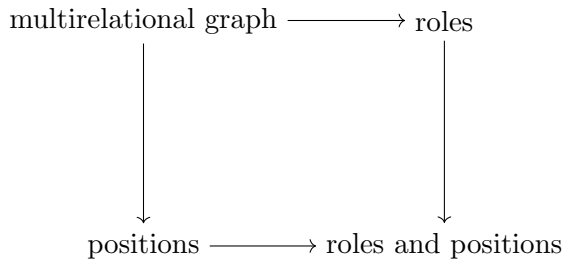

 \section{Perfect equivalences}

To motivate our framework, we first consider the case in which the nodes in a block are perfectly equivalent. 
Given a multirelational graph $G=(V,\{A_i\}_{i=1}^r)$ and a partition $B_1,\dots , B_m$ of its set of nodes from
positional analysis, we assume that all nodes in a block are structurally equivalent. Therefore, as we remarked in Section \ref{SS:density matrices}, we have that the density matrices $\widetilde{A_1},\dots \widetilde{A}_r$ are Boolean matrices, so if we set $\widetilde{V}=\{B_1,\dots , B_m\}$, it follows that $\left(\widetilde{V}, \{\widetilde{A}_i\}_{i=1}^r\right)$ is a multirelational graph. Given a vertex $i\in V$, let $\phi(i)$ denote the block in the partition that includes vertex $i$. The map $i\mapsto \phi(i)$ then induces a  graph homomorphism 
 \[
	 \left(V,\{A_i\}_{i=1}^r\right) \twoheadrightarrow \left(\widetilde{V}, \{\widetilde{A}_i\}_{i=1}^r\right)
 \]
that is surjective both at the level of vertices and edges. If there 
 is a tie from $i$ to $j$ under the relation $s$, then there is a $1$ in the $i$th row and $j$th column of  $A_s$. 
 Because the density matrix $\widetilde{A}_s$ is Boolean, there is also a $1$ in the entry that corresponds to the $\phi(i)$th row and $\phi(j)$th column of $\widetilde{A}_s$. Consequently, there is a tie from $\phi(i)$ to $\phi(j)$ under the relation $s$.  

We employ
the following definition.


\begin{definition} Let  $G=(V,\{A_i\}_{i=1}^r)$ be a multirelational graph, and let $\widetilde{V}=\{B_1,\dots , B_m\}$ be a partition of $V$ such that the density matrices $\widetilde{A}_1,\dots , \widetilde{A}_r$ are Boolean. A \define{reduction} $\RR(G)$ of $G$ is the multirelational graph $\left(V, \{\widetilde{A}_i\}_{i=1}^r\right)$ together with the graph homomorphism $\phi$ that is given by sending every node of $V$ to its block in the partition.
 Abusing notation, we often 
 denote the graph by  $\RR(G)$.
 \end{definition}

If we consider the semigroups of relations associated with a multirelational graph $G$ and with a reduction $\RR(G)$, we have the following lemma.
 
 \begin{lemma}\label{L:reduction homo}
 
 The  homomorphism $\phi\colon G\to \RR(G)$ induces a semigroup homomorphism 
 \[
	\SG(\phi)\colon  \SG(A_1,\dots , A_r)\to \SG(\widetilde{A}_1,\dots , \widetilde{A}_r) \,,
 \]
which is defined by sending each matrix $A_i$ to its density matrix $\widetilde{A}_i$. 
\end{lemma}

\begin{proof}\footnote{We note that we could have omitted this proof, as this lemma is a trivial example of what is discussed in  \cite{KR84}. However, we have decided to include a proof, as we think that it can help in illustrating the methods that we discuss.}
To see that this is indeed a homomorphism of semigroups, we need to show that it sends a product $A_iA_j$ to the product of the corresponding density matrices $\widetilde{A}_i \widetilde{A}_j$.  
A partition of the nodes into $m$ blocks induces block decompositions on the matrices $A_1,\dots , A_r$ that are compatible with each other, and therefore we can perform the products $A_iA_j$ as products of block matrices \cite[Theorem 1.9.6]{Ev80}. 
We denote the $(k,l)$th block of a matrix $A_i$ by $A_i^{kl}$. Then the $(k,l)$th block in the matrix $A_iA_j$ is given by the Boolean sum $A_i^{k1}A_j^{1l}+A_i^{k2}A_j^{2l}+\dots +A_i^{km}A_j^{ml}$. Further, we have that a product of blocks $A_i^{kh}A_j^{hl}$ is a block with all entries $1$ exactly when both blocks $A_i^{kh}$ and $A_j^{hl}$ have all entries $1$, and $0$ otherwise.  And finally the  $(k,l)$th block of  $A_iA_j$ contains all $1$s if at least one of the block  $A_i^{k1}A_j^{1l},A_i^{k2}A_j^{2l}, A_i^{km}A_j^{ml}$ contains all $1$s. The  $(k,l)$th block of  $A_iA_j$ contains all $0$s  otherwise.

Now, by the definition of density matrices, there is a $1$ (respectively, a $0$) in the $(k,h)$th entry of matrix $\tilde{A_i}$ exactly when the  $(k,h)$th block in the  matrix $A_i$ consists of all $1s$ (respectively $0$s). Thus, we see that there is a $1$ (respectively $0$), in the $(k,l)$th entry of the product  $\tilde{A_i}\tilde{A_j}$ exactly when at least one of the products 
$\tilde{A_i}^{k1}\tilde{A_j}^{1l}, \dots , \tilde{A_i}^{km}\tilde{A_j}^{ml}$ is $1$. Thus, we see that there is a $1$ (respectively $0$) in the $(k,l)$th entry of  $\tilde{A_i}\tilde{A_j}$ exactly when the $(k,l)$th block of $A_iA_j$ consists of all $1$s (respectively, $0$), and we obtain our claim, namely that 
$\SG(\phi)(A_i A_j)=\SG(\phi)(A_i)\SG(\phi)(A_j)$.

\end{proof}

Thus, by Lemma \ref{L:reduction homo} once we choose a reduction of $G$, we automatically also have a reduction of the corresponding semigroup of relations. 

As we discussed in Section \ref{S:intro}, ideally in positional analysis one is interested in studying not one blockmodel, but a nested sequence of blockmodels. 
n the ideal (yet unrealistic) case that all of these partitions have Boolean density matrices, one has a sequence $G_1,\dots , G_p$ of multirelational graphs and surjective multirelational graph homomorphisms $G\twoheadrightarrow G_i$ for each $i$. 
Because the different partitions are nested within each other, there are also surjective multirelational graph homomorphisms between the reduction graphs, and (up to a relabeling of the indices of the reduction) one gets a sequence $G=G_0,G_1,\dots , G_p$ of multirelational graphs and multirelational graph homomorphisms $\pi_{i,j}\colon G_i\twoheadrightarrow G_j$ for $i\leq j$ such that $\pi_{k,j}\circ \pi_{i,k}=\pi_{i,j}$ whenever $i\leq k \leq j$. If we now associate each multirelational graph $G_i$ with its semigroup of relations $\SG(G_i)$, we see
that each multirelational graph homomorphism $\pi_{i,j}$ also induces a homomorphism $\SG(\pi_{i,j})\colon \SG(G_i)\to \SG(G_j)$ of semigroups in a way that preserves composition of homomorphisms. That is, $\SG(\pi_{i,j})=\SG(\pi_{k,j})\circ \SG(\pi_{i,k})$ whenever $i\leq k \leq j$.

In the language of category theory, we see role analysis is a functor.  We devote Section \ref{S:cat and fun} and \ref{S:role analysis} to making this statement precise: we first introduce a few necessary
ideas from category theory, we then  prove that the assignment of the semigroup of relations to a multirelational graph, at least under the assumption of perfect equivalences, is a functor, and finally we explain why functoriality is important.

\section{Categories and functors}\label{S:cat and fun}

Category theory is a field that studies general abstract structures in mathematics. It does so in focusing not on the mathematical objects themselves, but in how they are related to each other in a way that preserves some of the structure of the objects.
 Therefore, a category consists of a certain family of objects, a family of morphisms between them that captures some structure-preserving way of relating the objects, and a rule to compose these morphisms. For example, one can have a category with objects that consist of graphs and morphisms that consist of graph homomorphisms, together with the usual composition of graph homomorphisms.

%
%
%
%
Following the same principle, one is interested not only in categories, but also in ways to relate them while preserving some of their structure. Consequently, one introduces the notion of a functor between categories: given two categories, $C$ and $D$, a functor $F\colon C\to D$  is given by a pair of maps --- one on the set of objects and one on the set of morphisms --- such that one preserves the composition of morphisms.\footnote{One can take this even further by considering structure-preserving ways to relate different functors to each other; this leads to the concept of natural transformations, but we will not use these in the present work.}

There are numerous introductory texts with a view towards applications, see, e.g. \cite{FS18,TB18}.


\begin{definition}
A \define{category} $C$ is given by the following data:
\begin{itemize}
\item a set of objects $\ob C$
\item  for every pair of objects, $x$ and $y$, a set of morphisms $C(x,y)$
\item for every triple of objects $x,y,z$ composition maps
\begin{alignat}{2}
	 C(x,y)\times C(y,z) &\notag  \to C(x,z)\\
	  (f \, , \, g) &\notag \mapsto g\circ f 
\end{alignat}

\item for every object $x$
a map that assigns to $x$ its identity morphism
\[
	\{\star\}\to C(x,x)\colon \star \mapsto 1_x \,.
\]
\end{itemize}
This data satisfies axioms that make the 
composition of morphisms  associative and the identity morphism for every object the neutral element for this composition:
 for any quadruple $x,y,z,w$ of objects and any triple $f\in C(x,y)$, $g\in C(y,z)$, $h\in C(z,w)$ of morphisms, we have that $h\circ (g\circ f)=(h\circ g)\circ f$. Furthermore, for any pair of objects, $x$ and $y$, and any morphism $f\in C(x,y)$, we have that $1_y\circ f=f$ and $f\circ 1_y=f$. In other words, the identity morphisms are left-identities and right-identities for the composition of morphisms.
 \end{definition}
 
 
\begin{definition}

Given two categories, $C$ and $D$, a \define{functor} from $C$ to $D$ is given by the following data:
\begin{itemize}
\item a map $F\colon \ob C\to \ob D$
\item a map $C(x,y)\to C(F(x),F(y))$ for every pair of objects, $x$ and $y$, of $C$. 
\end{itemize}
These maps satisfy axioms that ensure  that the composition of morphisms is preserved and identity morphisms are sent to identity morphisms.
\end{definition}
  In an abuse of notation, it is common to denote the functor, the map on the set of objects, and the maps on the sets of morphisms using the same symbol. For example, we use $F$ to denote the functor $F\colon C\to D$, the map $F\colon \ob C\to \ob D$, and the map $F\colon C(x,y)\to C(F(x),F(y))$. With this notation, we can express the two properties that these maps satisfy in a very simple way. Specifically, one requires that
 \[
 	F(g\circ f)=F(g)\circ F(f)
 \]
  for any pair of composable morphisms in $C$ and that
  \[
 	 F(1_x)=1_{F(x)} \, .
  \]

 
 \section{Functoriality of the semigroup of relations assignment}\label{S:role analysis}
We are now ready to state the  result that will motivate the new methods we introduce. 
Let $\krelGraph$ denote the category with objects $k$-relational graphs $(V, \{A_i\}_{i=1}^k)$ and morphisms  $k$-relational graph homomorphisms that are surjective at both the level of vertices and edges. Let $\semigroup$ denote the category with objects semigroups and morphisms surjective semigroup homomorphisms. We have:

 \begin{proposition}\label{T:role}
The assignment of the semigroup of relations to a multirelational graph induces a functor $R\colon \krelGraph \to \semigroup$.
 \end{proposition}

 \begin{proof}
 Given a $k$-relational graphs $G=(V, \{A_i\}_{i=1}^k)$ with $|V|=n$, we assign to it the semigroup of relations $R(G)=\SG(\{A_i\}_{i=1}^k)$. Now, given a surjective $k$-relational graph homomorphism $\phi\colon G\to G'$ with $G'=(V', \{A'_i\}_{i=1}^k)$, the preimages of the nodes in $V'$ under $\phi$ give a partition of the nodes in $V$, and thus a partition of the matrices $A_1,\dots , A_k$ into blocks. 
 Since the homomorphism $\phi$ is surjective on edges, we have that there is a $1$ (respectively, a $0$) in the $(k,l)$th entry of matrix $A'_i$ exactly when the $(k,l)$th block of $A_i$ consists of all $1$s (respectively, $0$s). 
 Thus, using block multiplication of matrices, analogously as in  the proof of  Lemma \ref{L:reduction homo}, we see that $\phi$ induces a semigroup homomorphism  $R(G)\to R(G')$. 
 
 It is easy to see  that the assignment $(\phi\colon G\to G')\mapsto \big(R(\phi)\colon R(G)\to R(G')\big)$  satisfies the properties of a functor, namely that it preserves composition of morphisms and identity morphisms. 
 
 \end{proof}

  \subsection{Why functoriality is important}

As we discussed at the beginning of Section \ref{S:cat and fun}, at the core of category theory is the idea that the study of mathematical objects should involve also the study of structure-preserving relationships between the objects. While category theory has been mostly used to study problems in abstract mathematics, in recent years concepts and ideas from category theory have found more and more applications to real-world problems.\footnote{See, for instance, the discussion in \cite{C09}, from which we quote the following excerpt: ``Functoriality has proven itself to be a powerful tool in
the development of various parts of mathematics, such as Galois theory within algebra, the theory of Fourier series within harmonic analysis, and the applicaton of
algebraic topology to fixed point questions in topology. We argue that [\dots]  it has a role to play in the study of point cloud data as well [\dots].''} 

One of the main advantages that the categorical approach brings to applications is that it gives a framework conducive to  studying stability questions: 
one is often only interested in studying a data set up to some notion of perturbation or noise, and thus it is not enough to 
associate invariants or algebraic structures to a single data set, but rather one seeks to study how, when changing the data sets, the corresponding algebraic structure or invariant changes. See, e.g., \cite{LFH10} for a discussion of robustness of centrality measures in network science.

In particular, in the current work, thanks to the functoriality of the assignment of the semigroups of relations to a multirelational graph,
we can tie role and positional analysis together, as shown in Fig.\ \ref{F: role functor}: the choice of a reduction of a multirelational graph $G$ automatically gives a choice of homomorphic reduction for its corresponding semigroup of relations, and furthermore, given several reductions of a multirelation graph, as for instance provided by the nested stochastic blockmodel (see Section \ref{S:SBM}), we have a way of relating the corresponding semigroups of relations, and thus we can tie together  role and positional analysis  at different hierarchical levels, and we can relate the analyses performed at different hierarchical levels to each other.

\begin{figure}
\begin{tikzpicture}
\node (1) at (0,0) {$G=\left(n,\{A_i\}_{i=1}^k\right)$};
\node (2) at (0,-2.5) {$H=\left(m,\{B_i\}_{i=1}^k\right)$};
\node (3) at (6,0) {$R(G)=\SG\left(\{A_i\}_{i=1}^k\right)$};
\node (4) at (6,-2.5) {$R(H)=\SG\left(\{B_i\}_{i=1}^k\right)$};
\draw[|->]
(1) edge node[above] {$R$} (3)
(2) edge node[above] {$R$} (4);
\draw[->>]
(1) edge node[left] {$\phi$} (2);
\draw[->]
(3) edge node[right]{$R(\phi)$}(4);
 \end{tikzpicture}
 
 \caption{Functoriality of role analysis allows us to tie together role analysis and positional analysis. }\label{F: role functor}
 \end{figure}

 \section{Approximate equivalences}

 There are two main issues in extending Proposition \ref{T:role} to the case of approximate equivalences:
\begin{enumerate}[label=(\roman*{})]
\item Given a blockmodel of a multirelational graph $G$, how do we obtain a multirelational graph $G'$ with nodes representing the blocks? Specifically, how do we define edges in $G'$?
\item Does any reduction $G\to G'$ as obtained in the previous point induce a semigroup homomorphism $\SG(G)\to \SG(G')$?
\end{enumerate}

A common way to address point (i) is to define an image matrix with a $1$ corresponding to each non-zero block, and a $0$ to each zero block. This method is often called ``lean fit'' in the literature \cite{BBA75}. The conditions under which such a reduction of the multigraph induces a semigroup homomorphism have been studied in \cite{KR84}.

Here  we seek a way to address these two questions in a way that does not depend on specific assumptions on the  blockmodel that one chooses. To do this, we consider multigraphs with weights associated to edges, and thus instead of working with image matrices, we decide to work directly with density matrices.  This leads to the additional issue that we thus now need to work with semigroups of relations that are generated by matrices that are not Boolean, and thus the associated semigroups might not be finite. However, when studying compound relations, one is only interested in studying a finite number of compound relations \cite{LW71}. This is a fundamental observation, which  leads us to introduce truncated semigroups of relations as a new method to study compound relations in social networks.

\subsection{Reductions of multigraphs with weighted edges}

\begin{definition}
A \define{weighted multirelational graph} $G$ consists of a pair $\left(V,\{A_i\}_{i=1}^r \right)$, where $V$ is a finite set of \define{vertices} of $G$
and the adjacency matrices $A_1,\dots , A_r$ are 
$n\times n$ matrices, with  entries in $[0,1]$, that we call \define{weighted relation matrices} of $G$.
\end{definition}

Now, given a relational multigraph $\left(V,\{A_i\}_{i=1}^r \right)$, and a blockmodel $B=\{B_1,\dots , B_m\}$ the  corresponding density matrices (see Definition \ref{D:density matrix})  are $m\times m$ weighted relation matrices $\widetilde{A_1},\dots , \widetilde{A_r}$  and there is a  multirelational graph homomorphism
\[
\phi\colon \left(V,\{A_i\}_{i=1}^r \right) \longrightarrow \left(B,\{\widetilde{A}_i\}_{i=1}^r \right)
\]
defined by sending node $i$ to the block $B_k$ to which it belongs, which is surjective both at the level of vertices as well as edges.

 \subsection{Truncated semigroups of relations}
 In most types of social networks, one is only interested in studying compound relations up to a certain length. As Lorrain and White write in \cite{LW71}: ``Only in kinship networks or in formal hierarchies are very long `words' known to
represent a meaningful type of relation.''  Thus, here we introduce the truncated semigroups $SG_k(G)$ associated to a given multirelational graph $G$, for $k=1,2,3,\dots$, where $SG_k(G)$  is generated by all $j$-fold products of the relation matrices, for $j\leq k$, and we impose the relations given by setting any $j'$-fold product for $j'>k$ equal to zero. One can consider the semigroups $SG_1(G), SG_2(G),\dots $ as giving a hierarchy of compound relations; the lower the value of $k$, the more the study of compound relations focuses on local structure.

 \begin{definition}
 Given two $n\times n$ weighted relation matrices $A$ and $B$, we define the following product:
 \[
  A \bigtriangledown B (i,j)=
  \max_{k=1,\dots , n} \left\{ A(i,k)B(k,j)\right \} \, .
  \]
 \end{definition}
  This operation is clearly associative, and the identity matrix is a left and right identity element. 
 One can interpret this product as follows: we are interested in studying compound relations (or equivalently, paths of a certain length), where now to each compound relation there is a certain degree of uncertainty associated to it. Given two actors $i$ and $j$, in the case of Boolean matrices we were interested in knowing not how many compound relations of a certain type there are between $i$ and $j$, but merely whether there was one or none. Now similarly, given compound relations of the same type between two actors $i$ and $j$, we are interested in knowing what is the highest degree of certainty to which there is such a path.

 We note that for  relation matrices  $A$ and $B$, the product $\bigtriangledown$ is the Boolean matrix product.

 \begin{definition}
Given a weighted multirelational graph $(V, \{A_i\}_{i=1}^r)$, and a positive integer $k$, its \define{$k$-truncated semigroup of relations} is the semigroup that is generated by the matrices $A_1,\dots , A_r$, with binary operation $\bigtriangledown$, and relations given by setting all $k+1$-fold products equal to $0$. We denote such a semigroup by $\SG_k(A_1,\dots , A_r)$.
\end{definition}

 \section{Experiments}
 
 We illustrate our methods using two well-known social networks. 
 We first use the SBM introduced by Peixoto (see Section \ref{S:SBM}) to get a nested sequence of blockmodels, and then study truncations of the corresponding semigroups of relations. 
 
 \subsection{Nested stochastic blockmodels (SBM)}\label{S:SBM}
Nested  blockmodels were first introduced by Peixoto in \cite{Peixoto14}. The motivation behind this blockmodel is that traditionally the choice of a 
blockmodel might not be robust and one might thus not be able to distinguish structure of a network from noise. Peixoto thus proposes a nested blockmodel that gives a hierarchy of blockmodels of a given graph, and which is  obtained by considering the multigraph with nodes given by block and the same count of edges, and obtaining iteratively blockmodels with fewer block from such multigraphs. This leads to a nested hierarchy of blockmodels, where at increasing hierarchy levels the number of nodes decreases, until one eventually reaches a multigraph with just one block, and thus just one node. The edge count remains the same at every level.
Here we use the  implementation provided by Peixoto in \verb|graph-tool|
\cite{graph-tool}.
We note that this is a stochastic model, and thus different iterations might lead to different nested blockmodels.
\subsection{Sampson Monks}
This is a network describing social  relations between 18 monks in an American Monastery in the 1960's. Sampson studied their relationships during a 12-month period, and during this period there was a major conflict, which resulted in the departure of a large number of members. Here we use the relational matrices studied in \cite{BBA75}: these are two non-symmetric $18\times 18$ matrices, one summarising a series of positive feelings, which we denote by $P$ the other a collection of negative ones, which we denote by $N$.  

\subsection{Lazega Lawyers}
This network describes relationships between $71$ lawyers in a  US corporate law firm in the late 1980s. The relationships are described by three $71\times 71$ non-symmetric relational matrices: one that encodes coworking relationship, one for advice, and one for mentorship. The data was originally published in \cite{lazega}.


\subsection{Results}

\begin{figure}
(a) \includegraphics[scale=0.2]{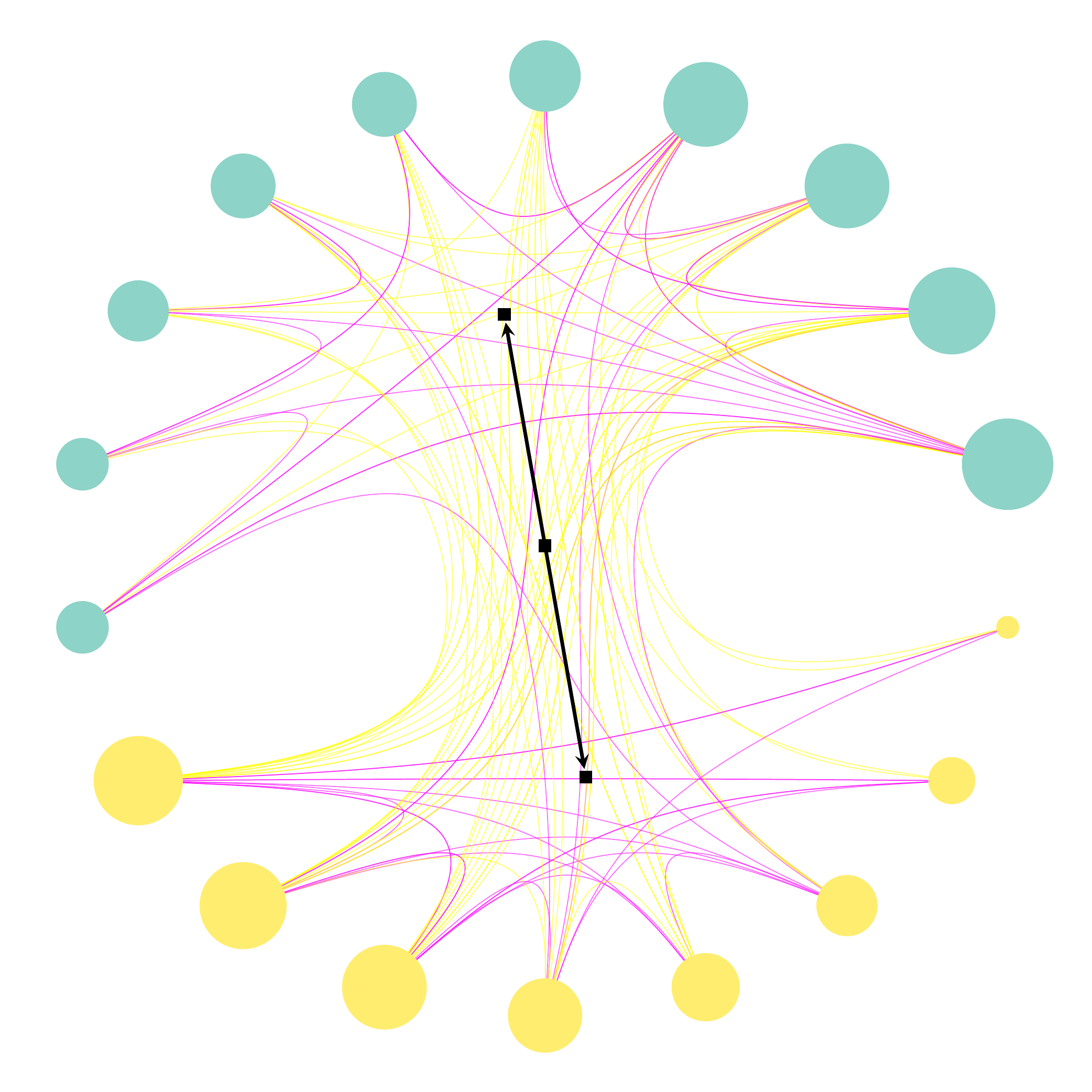} 
(b)\includegraphics[scale=0.3]{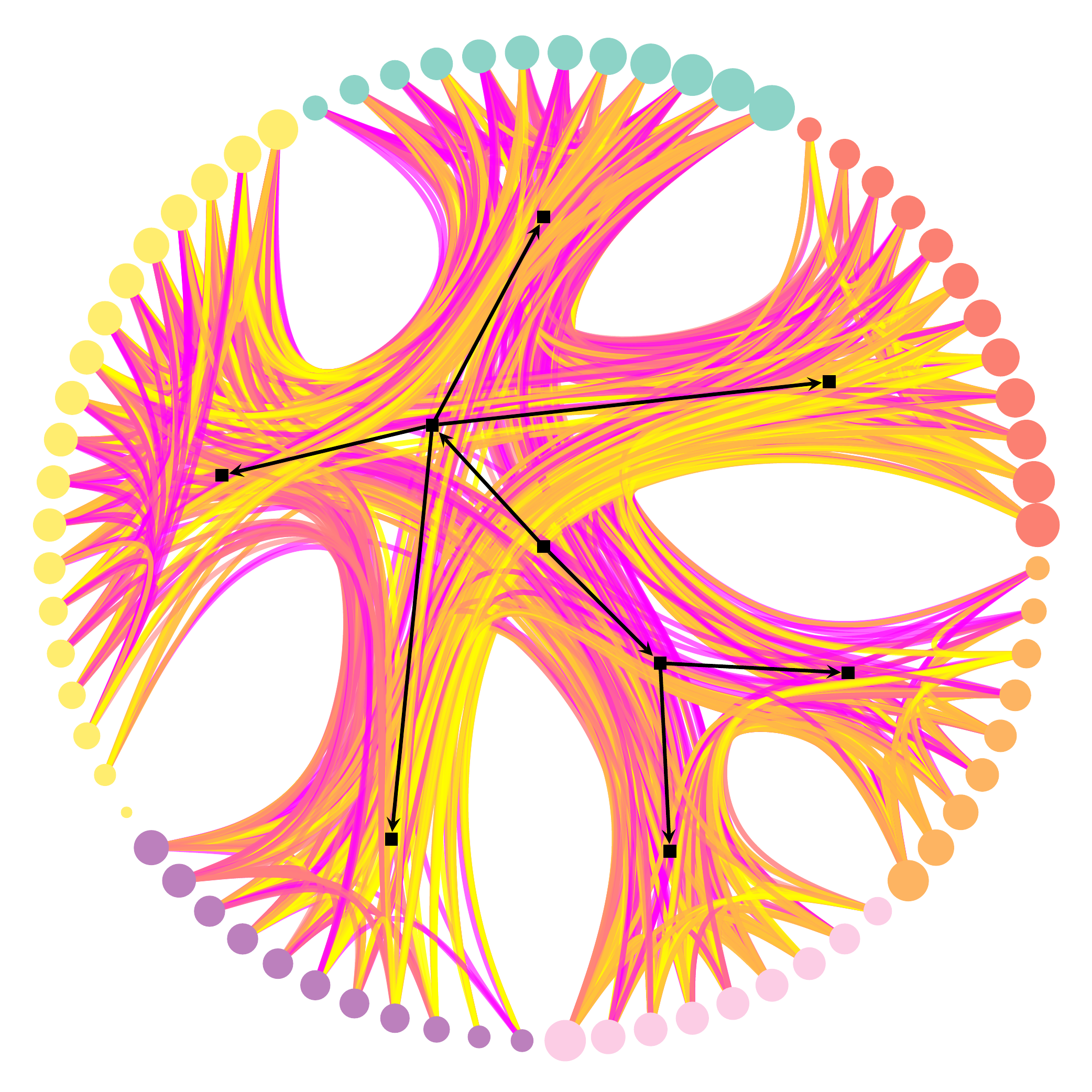}
\caption{The nested blockmodel structure of (a) the Sampson monks data set, where the algorithm returns a nested blockmodel with two levels; (b) the Lazega lawyer network, where the algorithm returns a nested blockmodel with three levels.
}\label{F:SBM}
\end{figure}

We discuss our methods for the datasets that we consider. 

\subsubsection{Monks}
The nested SBM returns two hierarchical levels for the Sampson Monks dataset: at the first level we have  the 18 nodes corresponding to the monks subdivided into two blocks, and at the next and final level the two nodes corresponding to the previous blocks are joined in the same block. We illustrate this subdivision in Figure \ref{F:SBM}(a). The subdivision returned by the SBM corresponds to the well-known subdivision studied in the literature between the ``loyal'' monks and the ``opposition''\cite{BBA75}. The semigroup of relations generated by $P$ and $N$  associated to the original network contains $55$ different elements. We next consider the density matrices $\widetilde{P}$ and $\widetilde{N}$ associated to these matrices. We have 
\[
\widetilde{P}=\begin{bmatrix}
0.11 & 0.25\\
0.11 & 0.25
\end{bmatrix}
\]
and 
\[
\widetilde{N}=\begin{bmatrix}
0.17 & 0.25\\
0.19 &0.20
\end{bmatrix}\, ,
\]
where we round the numbers up to two decimal digits.

We compute the truncated semigroup of relations $\SG_{18}(\widetilde{P},\widetilde{N})$, which contains eight elements. We note that already the semigroup $\SG_{4}(\widetilde{P},\widetilde{N})$ contains $8$ elements, and that indeed in this case the semigroup generated by the  density matrices $\widetilde{P}$ and $\widetilde{N}$ with rounded entries is finite. If we round the entries in the density matrices, we will indeed always obtain finite semigroups. Had we not rounded the matrix entries to two decimal digits, the truncated semigroup $\SG_{18}(\widetilde{P},\widetilde{N})$ would instead contain $4097$ elements. 
Apart from the two generators, we have the following $8$ compound relations.
All $4$-fold products are equal to the zero matrix:
\[\widetilde{N}\widetilde{N}\widetilde{N}\widetilde{N}=\dots = \widetilde{P}\widetilde{P}\widetilde{P}\widetilde{P}=
\begin{bmatrix}
0 & 0 \\
0 & 0
\end{bmatrix}\, .\]

For the $3$-fold products we have:

\[\widetilde{N}\widetilde{N}\widetilde{N}=\widetilde{P}\widetilde{N}\widetilde{N}= 
\widetilde{N}\widetilde{P}\widetilde{N}=
\widetilde{P}\widetilde{P}\widetilde{N}=
\widetilde{N}\widetilde{N}\widetilde{P}=
\widetilde{P}\widetilde{N}\widetilde{P}=
\begin{bmatrix}
0 .01& 0.01 \\
0.01 & 0.01
\end{bmatrix}\, ,
\]

\[\widetilde{N}\widetilde{P}\widetilde{P}=
\begin{bmatrix}
0 .01& 0.02 \\
0.01 & 0.01
\end{bmatrix}\]
and 
\[\widetilde{P}\widetilde{P}\widetilde{P}=
\begin{bmatrix}
0 .01& 0.02 \\
0.01 & 0.02
\end{bmatrix}\, .\]

The $2$-fold products are

\[\widetilde{N}\widetilde{N}=
\begin{bmatrix}
0 .05& 0.05 \\
0.04 & 0.05
\end{bmatrix}\]

\[\widetilde{P}\widetilde{N}=
\begin{bmatrix}
0 .05& 0.05 \\
0.05 & 0.05
\end{bmatrix}\]

\[\widetilde{N}\widetilde{P}=
\begin{bmatrix}
0 .03& 0.06 \\
0.02 & 0.05
\end{bmatrix}\]

\[\widetilde{P}\widetilde{P}=
\begin{bmatrix}
0 .03& 0.06 \\
0.03& 0.06
\end{bmatrix}\, .\]

\subsubsection{Lawyers}

The nested SBM returns a nested blockmodel with three levels: the first with $71$ nodes corresponding to the lawyers, divided into six blocks; the second with $6$ nodes, divided into two blocks, and the last with one block. We denote the corresponding density matrices for the first  and second level with subscripts $1$ and $2$, respectively. 
The semigroup $\SG(\Adv,\Fr,\Work)$ contains $202$ elements. The truncated semigroup $\SG_{15}(\Adv_1,\Fr_1,\Work_1)$ contains $66$ element, when we approximate matrix entries up to 2 decimal entries. We note that in this case the semigroup $\SG(\Adv,\Fr,\Work)$ contains $66$ elements, and we have
${\SG_4(\Adv_1,\Fr_1,\Work_1)=\SG_{15}(\Adv_1,\Fr_1,\Work_1)}$.

Next, we compute the truncated semigroups for the density matrices corresponding to the second level in the hierarchy. Here we note that since in the SBM the graphs at each level have the same number of nodes, we can again compute the entries in the matrices
$\Adv_2,\Fr_2,\Work_2$ as fraction of $1$s. 
We have $\SG_3(\Adv_2,\Fr_2,\Work_2)=\SG(\Adv_2,\Fr_2,\Work_2)$ when approximating matrix entries up to two decimal digits, and this semigroup contains $19$ elements.

%

\section{Conclusions and discussion}

In this paper we have introduced new methods to perform role and positional analysis in networks, by proposing weighted multirelational graphs, with weights capturing the level of uncertainty at which a certain relation between blocks of a blockmodel occurs, and truncated semigroups of relations to study patterns between such relations. Our framework is very general, and it does not depend on the algorithm chosen to find a blockmodel, nor on the definition of similarity that one uses. 

Some interesting follow-up questions include:
\begin{itemize}
\item extending this framework to multilayer networks;
\item exploiting our framework to study stability of positional and role analysis;
\item comparing truncated semigroups of relations obtained from different types of blockmodel algorithms.
 \end{itemize}
%
%
 
 
\bibliographystyle{plain}
\bibliography{equiv}
 
 
\end{document}